\documentclass[jmp,graphicx]{revtex4-1}
\usepackage{amsthm,epsfig,graphicx,amsmath,amssymb}
\usepackage{color}

\draft 

\newtheorem{claim}{Claim}[section]
\newtheorem{theorem}[claim]{Theorem}

\newtheorem{definition}[claim]{Definition}
\newtheorem{corollary}[claim]{Corollary}

\begin{document}


\title{Pseudo-orbit approach to trajectories of resonances in quantum graphs with general vertex coupling: Fermi rule and high-energy asymptotics} 



\author{Pavel Exner}
\email[]{exner@ujf.cas.cz}
\homepage[]{http://gemma.ujf.cas.cz/~exner/}
\thanks{P.E. was supported by the project 14-06818S ``Rigorous Methods in Quantum Dynamics:
Geometry and Magnetic Fields'' of the Czech Science Foundation}
\affiliation{Department of Theoretical Physics, Nuclear Physics Institute ASCR, 25068 \v{R}e\v{z} near Prague, Czechia}
\affiliation{Doppler Institute for Mathematical Physics and Applied Mathematics,
Czech Technical University in Prague, B\v{r}ehov\'a 7, 11519 Prague, Czechia}

\author{Ji\v{r}\'{i} Lipovsk\'{y}}
\email[]{jiri.lipovsky@uhk.cz}
\homepage[]{http://lide.uhk.cz/lipovji1/}
\thanks{J.L. was supported by the project 15-14180Y ``Spectral and Resonance Properties of Quantum Models'' of the Czech Science Foundation}
\affiliation{Department of Physics, Faculty of Science, University of Hradec Kr\'{a}lov\'{e},
Rokitansk\'eho 62, 500\,03 Hradec Kr\'{a}lov\'{e}, Czechia}


\date{\today}

\begin{abstract}
The aim of the paper is to investigate resonances in quantum graphs with a general self-adjoint coupling in the vertices and their trajectories with respect to varying edge lengths. We derive formul\ae\ determining the Taylor expansion of the resonance pole position up to the second order which represent, in particular, a counterpart to the Fermi rule derived recently by Lee and Zworski for graphs with the standard coupling. Furthermore, we discuss the asymptotic behavior of the resonances in the high-energy regime in the situation where the leads are attached through $\delta$ or $\delta_\mathrm{s}'$ conditions, and we prove that in the case of $\delta_\mathrm{s}'$ coupling the resonances approach to the real axis with the increasing real parts as $\mathcal{O}\big((\mathrm{Re\,}k)^{-2}\big)$.
\end{abstract}


\maketitle 

\section{Introduction}

The idea to investigate quantum mechanics of systems the configuration space of which is a metric graph dates back to the 1930's but it attracted a substantial attention only in the last three decades. Quantum graphs, from the mathematical point of view families of ODE's coupled by appropriate conditions, provide nowadays a simple model to numerous physical systems, small as well as large. For a thorough description of state of art in this field we refer the reader to the monograph \cite{BK} and the references therein.

Quantum graphs having semiinfinite edges are useful to study resonance effects. As it is well known, the notion of resonance can bear different meanings. One has, in particular, \emph{resolvent resonances} understood as poles of the meromorphic continuation of the resolvent to an `unphysical' sheet of the energy surface, or \emph{scattering resonances} understood as poles of the meromorphic continuation of the determinant of the scattering matrix. It was proven in \cite{EL1} that in the present context the two notions coincide, hence we further speak about resonances with adjectives. We note that sometimes it is useful to consider together with these `true' resonances also the eigenvalues, including embedded ones, with eigenfunctions supported only on the internal part of the graph \cite{DEL, DP}.

The said embedded eigenvalues give also often rise to (true) resonances when being exposed to a perturbation; this is one of the most common mechanisms behind this effect. A natural way to perturb a graph consists of varying lengths of its edges; such resonances originating from an embedded eigenvalue are sometimes called \emph{topological resonances}. They were studied recently in \cite{EL2, GSS, LZ}; in the last named paper Lee and  Zworski found a formula for the second derivative $\mathrm{Im\,}\ddot{k}$ with respect to the edge length in the case of the standard (Kirchhoff) coupling in the graph vertices; here $k$ is the square root of energy.

Since the standard coupling is not the only possibility, and in fact, \emph{every} self-adjoint vertex coupling can be given a physical meaning \cite{EP}, we investigate in this paper the corresponding Taylor expansion of the resonance pole positions in this more general setting. The method we use is based on pseudo-orbit expansion introduced recently in \cite{Li3}. We will derive the expression for $\mathrm{Im\,}\ddot{k}$ and $\mathrm{Re\,}\ddot{k}$ in the case graphs with a general vertex coupling. We will see that the real part of $\ddot{k}$ is needed in general to express the asymptotics of the resonances trajectories in the vicinity of the eigenvalue; instead of the parabolic curves $\mathrm{Im\,}k = c(\mathrm{Re\,}k)^2$ appearing, e.g., in Figure~4 in \cite{LZ} we obtain `slanted' parabolas -- see Figure~\ref{fig1} below. The general result will be illustrated by two examples in Section~\ref{sec-ex1}.

In the second part of the paper we address a different question; we ask about the high-energy asymptotic behavior of resonances in the situation when the leads are attached to the inner part of the graph by either $\delta$ or $\delta_\mathrm{s}'$ coupling while the other vertex conditions may be in general arbitrary self-adjoint. In the $\delta$ case we find that the resonances are asymptotically distributed as the ones for the standard coupling. On the other hand, the resonances for $\delta_\mathrm{s}'$-coupling converge in the high-energy limit to the eigenvalues of a decoupled graph with Neumann conditions and their imaginary parts behave as $(\mathrm{Re\,}k)^{-2}$. We again illustrate this result on simple examples in the last section.

\section{Preliminaries}

Suppose we have a metric graph $\Gamma$ which consists of a set of vertices $\mathcal{V}=\{\mathcal{X}_j\}$, a set of $N$ internal edges $\mathcal{E}_\mathrm{i}=\{e_j\}_{j=1}^N$ which we parametrized identifying them with intervals $(0,\ell_j)$, and a set of $M$ external semiinfinite edges $\mathcal{E}_\mathrm{e}=\{e_j\}_{j=N+1}^{N+M}$ identified with the halflines $(0,\infty)$. The graph is equipped with the second order differential operator $H$ acting as $-\frac{\mathrm{d}^2}{\mathrm{d}x^2}$. The latter is regarded as Hamiltonian of a quantum particle provided we use the units in which the mass of the particle and the reduced Planck constant satisfy $\hbar = 2m = 1$. To make it a self-adjoint operator we have to choose the domain properly; this is achieved by restricting it to functions in $W^{2,2}(\Gamma) = \oplus_{j=1}^{N+M} W^{2,2}(e_j)$ which fulfill the coupling conditions
\begin{equation}
  (U_j-I)\Psi_j + i (U_j+I)\Psi_j' = 0  \label{eq-cc}
\end{equation}
at the vertices, where $U_j$ is $d\times d$ unitary matrix, $d$ is the degree of the vertex, $I$ is the $d\times d$ identity matrix, $\Psi_j$ is the vector of limits of the functional values from the edges to the $j$-th vertex, and $\Psi_j'$ is the vector of limits of the outgoing derivatives.

A useful trick when dealing graphs having external edges is to consider the inner part only replacing the leads by an effective energy-dependent coupling $\tilde U_j(k)$ at the vertices where the leads were attached. Let a vertex connect $n$ internal edges and $m$ halflines and let the corresponding matrix $U_j$ consist of blocks $U_j = \left(\begin{array}{cc}U_1& U_2\\ U_3 & U_4\end{array}\right)$, where the $n\times n$ matrix $U_1$ corresponds to the coupling between the internal edges, $m\times m$ matrix $U_4$ between the external edges and the rectangular matrices $U_2$ and $U_3$ to the mixed coupling. The matrix $\tilde U_j(k)$ has the form \cite{EL2}
$$
  \tilde U_j(k) = U_1 - (1-k)U_2[(1-k)U_4-(k+1)I]^{-1}U_3
$$
and the boundary value vectors $\tilde \Psi_j$ and $\tilde \Psi_j'$ with $n$ entries corresponding to the internal edges satisfies the condition
$$
  (\tilde U_j-I)\tilde \Psi_j + i (\tilde U_j+I)\tilde \Psi_j' = 0\,,
$$
There are several classes of vertex conditions which will play a role later in the paper and it is useful to list them now:

\begin{itemize}
\item {\bf $\delta$-conditions} are in the $(n+m)\times (n+m)$-parameter family the only ones with the wave functions continuous at the vertex. They can be written as
\begin{eqnarray*}
  f (\mathcal{X}) & \equiv & f_i (\mathcal{X})=  f_j (\mathcal{X})\quad \hbox{for all }i,j\in \{1,\dots , n+m\}\\
  \sum_{j = 1}^{n+m}   f'_j (\mathcal{X}) & = & \alpha   f (\mathcal{X})
\end{eqnarray*}
corresponding to the unitary matrix is $U = \frac{2}{n+m+ i \alpha} J - I$, where $J$ is $(n+m)\times (n+m)$ matrix with all the entries equal to one.

\item {\bf $\delta_{\mathrm{s}}'$-conditions}  have the roles and the functions and derivatives interchanged:
\begin{eqnarray*}
  f ' (\mathcal{X}) & \equiv & f'_i(\mathcal{X})=  f'_j (\mathcal{X})\,,\quad \hbox{for all }i,j\in \{1,\dots , n+m\}\\
  \sum_{j = 1}^{n+m}   f_j (\mathcal{X}) & = & \beta   f' (\mathcal{X})\,.
\end{eqnarray*}
They can be regarded as a generalization of $\delta'$ condition on the line which preserves the permutation symmetry; the coupling matrix is $U = I - \frac{2}{n+m -i \beta} J$ in this case.
\item {\bf standard conditions} (sometimes called Kirchhoff, free, or even Neumann) represent a special case of $\delta$-condition for $\alpha = 0$, i.e. the functional values are continuous at the vertex and the sum of the outgoing derivatives vanishes. The corresponding unitary matrix is $U = \frac{2}{n+m} J - I$.
\item {\bf Dirichlet conditions} means that all the functional values are zero at the vertex, the unitary matrix is $U = -I$.
\item {\bf Neumann conditions}, on the other hand, mean that all the derivatives vanish at the vertex, and $U$ is now the identity matrix.
\end{itemize}

To conclude the preliminaries, let us make precise the concept of resonance mentioned in the introduction.

\begin{definition}
By the a \emph{resonance} we mean a complex $k^2$ for which there exists a generalized eigenfunction in $L^2_{\mathrm{loc}}(\Gamma)$, $f\not \equiv 0$, which satisfies the Schr\"odinger equation $-f''(x) + k^2 f(x) = 0$ on all the edges together with the coupling conditions (\ref{eq-cc}) at the vertices, and on all the external edges it has the form $c_j \,\mathrm{e}^{ikx}$.
\end{definition}

\section{Pseudo-orbit expansion for the resonance condition}

In this section we review the method of pseudo-orbit expansion for the resonance condition, which was developed recently in \cite{Li3}; for illustrating examples see \cite{Li4}. First, we define the effective vertex-scattering matrix $\tilde \sigma^{(v)}(k)$. Suppose again that we have a vertex $v$ which connects $n$ internal and $m$ external edges. All the internal edges are parametrized by $(0,\ell_j)$ with the point $x = 0$ corresponding to $v$; the leads are parametrized by $(0,\infty)$ with $x = 0$ corresponding to $v$. The generalized eigenfunctions components are $f_j(x) = a_j^{\mathrm{in}}\,\mathrm{e}^{-ikx}+a_j^{\mathrm{out}}\,\mathrm{e}^{ikx}$, $j = 1,\dots, n$, on the internal edges and $g_s(x) = b_s\,\mathrm{e}^{ikx}$, $s = 1,\dots, m$, on the external edges.

\begin{definition}
The effective vertex-scattering matrix $\tilde \sigma^{(v)}(k)$ is the $n\times n$ matrix which maps the vector of coefficients of the amplitudes of the incoming waves into the vector of the amplitudes of the outgoing waves $\mathbf{a}_v^{\mathrm{out}} = \tilde \sigma^{(v)} \mathbf{a}_v^{\mathrm{in}}$, where $\mathbf{a}_v^{\mathrm{in}} = (a_1^{\mathrm{in}}, \dots, a_n^{\mathrm{in}})^\mathrm{T}$ and  $\mathbf{a}_v^{\mathrm{out}} = (a_1^{\mathrm{out}}, \dots, a_n^{\mathrm{out}})^\mathrm{T}$.
\end{definition}

There is a simple connection between the matrices $\tilde \sigma(k)$ and $\tilde U(k)$, see Theorem 4.2 in \cite{Li3}. For simplicity, we drop the subscript $v$ specifying the vertex in question.

\begin{theorem}
The effective vertex-scattering matrix is $\tilde\sigma (k) = - [(1-k)\tilde U(k)-(1+k)I_n]^{-1}[(1+k)\tilde U(k)-(1-k)I_n]$, where $I_n$ is $n\times n$ identity matrix.
\end{theorem}

Next we introduce matrices we shall need to state the resonance condition. To this purpose we associate with the graph $\Gamma$ an oriented graph $\Gamma_2$ in which the external edges of $\Gamma$ are removed and each internal edge is replaced by two oriented edges (conventionally called \emph{bonds}) $b_j$, $\hat b_j$ of the lengths $\ell_j$ and opposite orientations. We use a permutation, not unique of course, between the list of bonds grouped according to the vertices to which they are directed and that that grouped by their orientation.

\begin{definition}
For a fixed permutation indicated above, we define the energy-dependent $2N\times 2N$ matrix $\tilde{\Sigma}(k)$ as a matrix similar to the block diagonal matrix with blocks $\tilde \sigma_v (k)$, the similarity transformation being determined by the bijective map between the bases
$$
\vec{\alpha} = (\alpha_{b_1}^\mathrm{in},\dots, \alpha_{b_N}^\mathrm{in},\alpha_{\hat{b}_1}^\mathrm{in},\dots,
\alpha_{\hat{b}_N}^\mathrm{in})^\mathrm{T}
$$
and
$$
(\alpha_{b_{v_{1}1}}^\mathrm{in},\dots,\alpha_{b_{v_{1}d_1}}^\mathrm{in},\alpha_{b_{v_{2}1}}^\mathrm{in},\dots,\alpha_{b_{v_{2}d_2}}^\mathrm{in},\dots)^\mathrm{T}\,,
$$
where $b_{v_{1}j}$ is the $j$-th edge with the endpoint at the vertex $v_1$. We also introduce three other $2N\times 2N$ matrices, the matrix $Q = \left(\begin{array}{cc}0& I_N\\I_N & 0\end{array}\right)$, the scattering matrix $S(k) = Q \tilde\Sigma (k)$, and
$$
  L =\mathrm{diag\,}(\ell_1,\dots , \ell_N,\ell_1,\dots , \ell_N)\,.
$$ 
\end{definition}

\noindent These notions allow us to state the resonance condition for the proof of which we refer to Theorem 4.5 in \cite{Li3}.

\begin{theorem} \label{thm: res1}
Resonances of the graph in question are given a solutions to the equation
$$
  \mathrm{det\,}(\mathrm{e}^{ikL} Q \tilde\Sigma(k)- I_{2N}) = 0\,.
$$
\end{theorem}

\noindent Next we reformulate this resonance condition using pseudo-orbits. To begin with, we again need some definitions.

\begin{definition}
A \emph{periodic orbit} $\gamma$ on the graph $\Gamma_2$ is a closed path which starts and ends at the same vertex. We denote it using the involved subsequent bonds, $\gamma = (b_1, b_2, \dots, b_n)$ noting that a cyclic permutation of bonds does not change the orbit. A \emph{pseudo-orbit} is a collection of periodic orbits ($\tilde \gamma = \{\gamma_1, \gamma_2, \dots ,\gamma_m \}$). An \emph{irreducible pseudo-orbit} $\bar \gamma$ is a pseudo-orbit, which contains no bond more than once. The \emph{metric length} of a periodic orbit is defined as $\ell_\gamma = \sum_{b_j\in \gamma} \ell_{b_j}$; the length of a pseudo-orbit is the sum of the lengths of all periodic orbits from which it is composed. By $A_{\gamma}$ we denote the product $S_{b_2 b_1} S_{b_3 b_2} \dots S_{b_1 b_n}$ of scattering amplitudes along the periodic orbit $\gamma = (b_1,b_2, \dots b_n)$; here $S_{b_i b_j}$ denotes the entry of the matrix $S$ in the row corresponding to the bond $b_i$ and column corresponding to the bond $b_j$. For a pseudo-orbit we define $A_{\tilde \gamma} = \prod_{\gamma_j\in \tilde\gamma}A_{\gamma_j}$. By $m_{\tilde \gamma}$ we denote the number of periodic orbits in the pseudo-orbit $\tilde \gamma$. By definition, the set of irreducible pseudo-orbits contains also irreducible pseudo-orbit on zero bonds with $m_{\bar \gamma} = 0$, $\ell_{\bar\gamma} = 0$ and $A_{\bar\gamma} = 1$.
\end{definition}

\noindent Armed with these notions we can formulate the theorem on the resonances in terms of the pseudo-orbits which was stated as Theorem~4.7 in \cite{Li3} and the proof of which followed from Theorem~1 in \cite{BHJ}.

\begin{theorem}
The resonance condition of Theorem~\ref{thm: res1} can be restated as
\begin{equation}
  \sum_{\bar \gamma} (-1)^{m_{\bar \gamma}} A_{\bar \gamma}(k) \,\mathrm{e}^{ik\ell_{\bar \gamma}} = 0\,,\label{eq-rescon}
\end{equation}
where the sum runs over all the irreducible pseudo-orbits $\bar \gamma$.
\end{theorem}

\section{Resonance behavior near the eigenvalue}

Let us consider a quantum graph with a general vertex coupling and attached halflines as described above. The resonance condition is given by the pseudo-orbit expansion as (\ref{eq-rescon}); we denote its left-hand side by $F(k(t),t)$ where $t$ will be the parameter governing variation of the graph edge lengths.

\begin{theorem} \label{thm-main}
Let the internal graphs edge lengths $\ell_j = \ell_j(t)$ depend on the parameter $t$ as $C^2$ functions. Suppose that at least some of them are non-constant in the vicinity of $t=0$ and that at that point the system has an eigenvalue $k_0^2>0$ embedded in the continuous spectrum. Then for small enough $t$ the condition (\ref{eq-rescon}) has a unique solution $k^2$, either an embedded eigenvalue or a resonance, and the following holds:
\begin{enumerate}
\item[(i)] $\dot k \in \mathbb{R}$, where dot signifies the derivative with respect to $t$.
\item[(ii)] Furthermore, we have
\begin{equation}
  \dot k \sum_{\bar\gamma}\left(\ell_{\bar\gamma}A_{\bar\gamma}(k)-i\frac{\partial A_{\bar\gamma}(k)}{\partial k}\right)(-1)^{m_{\bar\gamma}}\,\mathrm{e}^{ik\ell_{\bar\gamma}} + k \sum_{\bar\gamma}\dot \ell_{\bar\gamma} (-1)^{m_{\bar\gamma}} A_{\bar\gamma}(k) \,\mathrm{e}^{ik\ell_{\bar\gamma}} = 0\,, \label{eq-kdot}
\end{equation}
and
\begin{eqnarray}
  \ddot k \sum_{\bar\gamma}\left(\ell_{\bar\gamma}A_{\bar\gamma}(k)-i\frac{\partial A_{\bar\gamma}(k)}{\partial k}\right)(-1)^{m_{\bar\gamma}}\,\mathrm{e}^{ik\ell_{\bar\gamma}}
  +2 \dot k \sum_{\bar\gamma} \left(ik\ell_{\bar\gamma}\dot \ell_{\bar\gamma} A_{\bar \gamma}(k)+\dot\ell_{\bar\gamma}A_{\bar\gamma}(k)+k\dot\ell_{\bar\gamma}\frac{\partial A_{\bar\gamma}(k)}{\partial k}\right)(-1)^{m_{\bar\gamma}}\,\mathrm{e}^{ik\ell_{\bar\gamma}}  \nonumber\\
  \hspace{-7mm}+(\dot k)^2\sum_{\bar\gamma}\left(2\ell_{\bar\gamma}\frac{\partial A_{\bar\gamma}(k)}{\partial k}-i \frac{\partial^2 A_{\bar\gamma}(k)}{\partial k^2}+i \ell_{\bar\gamma}^2 A_{\bar\gamma}(k)\right)(-1)^{m_{\bar\gamma}}\,\mathrm{e}^{ik\ell_{\bar\gamma}}
  +k \sum_{\bar\gamma}(\ddot\ell_{\bar\gamma}+ik(\dot\ell_{\bar\gamma})^2)A_{\bar\gamma}(k)(-1)^{m_{\bar\gamma}}\,\mathrm{e}^{ik\ell_{\bar\gamma}} = 0\,. \label{eq-kddot}
\end{eqnarray}
\end{enumerate}
\end{theorem}
\begin{proof}
(i) By the implicit function theorem, the condition (\ref{eq-rescon}) is for small $t$ solved by
$$
  k = k_0 + \dot k t + \mathcal{O}(t^2)\,.
$$
Suppose that $\mathrm{Im\,}\dot k|_{k_0} = c$ holds with $|c|>0$. Then we have two possibilities, either the solution corresponds to a wavefunction localized on the inner part of $\Gamma$, or it has components on the leads. In the latter case, however, there is an $\varepsilon>0$ such that $\mathrm{Im}\,k>0$ for either $t\in(0,\varepsilon)$ or $t\in(-\varepsilon,0)$. Since $|k_0|>0$ and the wavefunction components on the leads are then square integrable, in both cases we conclude that there is a non-real eigenvalue at $k^2$ which contradicts, of course, to self-adjointness of the Hamiltonian. \\[.2em]
(ii)  The left-hand side of eq. (\ref{eq-kdot}) is obtained from (\ref{eq-rescon}) as $-i \frac{\mathrm{d} F(k(t),t)}{\mathrm{d} t}$, the relation (\ref{eq-kddot}) then follows by taking the derivative of (\ref{eq-kdot}) with respect to $t$.
\end{proof}

This result has the following corollary, which simplifies the computation of $\dot k$ and $\mathrm{Im\,}\ddot k$ for a certain family of coupling conditions which contains, in particular, the standard coupling.

\begin{corollary}
Let all the $A_{\bar\gamma}$ be $k$-independent and real. Then we have
$$
  \dot k = -k\, \frac{\sum_{\bar\gamma}\dot \ell_{\bar\gamma}(-1)^{m_{\bar\gamma}}A_{\bar\gamma}\cos{k\ell_{\bar\gamma}}}
  {\sum_{\bar\gamma}\ell_{\bar\gamma}(-1)^{m_{\bar\gamma}}A_{\bar\gamma}\cos{k\ell_{\bar\gamma}}}\,,
$$
\begin{eqnarray*}
  \mathrm{Im\,}\ddot k = -\frac{\sum_{\bar\gamma}\ell_{\bar\gamma}(-1)^{m_{\bar\gamma}}A_{\bar\gamma}\cos{k\ell_{\bar\gamma}}}{d}
  \left[2\dot k \sum_{\bar\gamma} (k \dot \ell_{\bar\gamma}\ell_{\bar\gamma}\cos{k\ell_{\bar\gamma}}+\dot \ell_{\bar\gamma}\sin{k\ell_{\bar\gamma}})A_{\bar\gamma}(-1)^{m_{\bar\gamma}} \right.\\
  \left.+(\dot k)^2\sum_{\bar\gamma}\ell_{\bar\gamma}^2(-1)^m_{\bar\gamma}A_{\bar\gamma}\cos{k\ell_{\bar\gamma}}+k \sum_{\bar\gamma}(k(\dot\ell_{\bar\gamma})^2\cos{k\ell_{\bar\gamma}}
  +\ddot\ell_{\bar\gamma}\sin{k\ell_{\bar\gamma}})(-1)^{m_{\bar\gamma}}A_{\bar\gamma}\right] \\
  +\frac{\sum_{\bar\gamma}\ell_{\bar\gamma}(-1)^{m_{\bar\gamma}}A_{\bar\gamma}\sin{k\ell_{\bar\gamma}}}{d}
  \left[2\dot k \sum_{\bar\gamma} (-k \dot \ell_{\bar\gamma}\ell_{\bar\gamma}\sin{k\ell_{\bar\gamma}}+\dot \ell_{\bar\gamma}\cos{k\ell_{\bar\gamma}})A_{\bar\gamma}(-1)^{m_{\bar\gamma}} \right.\\
  \left.-(\dot k)^2\sum_{\bar\gamma}\ell_{\bar\gamma}^2(-1)^m_{\bar\gamma}A_{\bar\gamma}\sin{k\ell_{\bar\gamma}}+k \sum_{\bar\gamma}(\ddot\ell_{\bar\gamma}\cos{k\ell_{\bar\gamma}}
  -k(\dot\ell_{\bar\gamma})^2\sin{k\ell_{\bar\gamma}})(-1)^{m_{\bar\gamma}}A_{\bar\gamma}\right]\,,
\end{eqnarray*}
\begin{eqnarray*}
  \mathrm{Re\,}\ddot k = -\frac{\sum_{\bar\gamma}\ell_{\bar\gamma}(-1)^{m_{\bar\gamma}}A_{\bar\gamma}\sin{k\ell_{\bar\gamma}}}{d}
  \left[2\dot k \sum_{\bar\gamma} (k \dot \ell_{\bar\gamma}\ell_{\bar\gamma}\cos{k\ell_{\bar\gamma}}+\dot \ell_{\bar\gamma}\sin{k\ell_{\bar\gamma}})A_{\bar\gamma}(-1)^{m_{\bar\gamma}} \right.\\
  \left.+(\dot k)^2\sum_{\bar\gamma}\ell_{\bar\gamma}^2(-1)^m_{\bar\gamma}A_{\bar\gamma}\cos{k\ell_{\bar\gamma}}+k \sum_{\bar\gamma}(k(\dot\ell_{\bar\gamma})^2\cos{k\ell_{\bar\gamma}}
  +\ddot\ell_{\bar\gamma}\sin{k\ell_{\bar\gamma}})(-1)^{m_{\bar\gamma}}A_{\bar\gamma}\right] \\
  -\frac{\sum_{\bar\gamma}\ell_{\bar\gamma}(-1)^{m_{\bar\gamma}}A_{\bar\gamma}\cos{k\ell_{\bar\gamma}}}{d}
  \left[2\dot k \sum_{\bar\gamma} (-k \dot \ell_{\bar\gamma}\ell_{\bar\gamma}\sin{k\ell_{\bar\gamma}}+\dot \ell_{\bar\gamma}\cos{k\ell_{\bar\gamma}})A_{\bar\gamma}(-1)^{m_{\bar\gamma}} \right.\\
  \left.-(\dot k)^2\sum_{\bar\gamma}\ell_{\bar\gamma}^2(-1)^m_{\bar\gamma}A_{\bar\gamma}\sin{k\ell_{\bar\gamma}}+k \sum_{\bar\gamma}(\ddot\ell_{\bar\gamma}\cos{k\ell_{\bar\gamma}}
  -k(\dot\ell_{\bar\gamma})^2\sin{k\ell_{\bar\gamma}})(-1)^{m_{\bar\gamma}}A_{\bar\gamma}\right]\,,
\end{eqnarray*}
where
$$
  d := \left({\sum_{\bar\gamma}\ell_{\bar\gamma}(-1)^{m_{\bar\gamma}}A_{\bar\gamma}\cos{k\ell_{\bar\gamma}}}\right)^2 + \left({\sum_{\bar\gamma}\ell_{\bar\gamma}(-1)^{m_{\bar\gamma}}A_{\bar\gamma}\sin{k\ell_{\bar\gamma}}}\right)^2\,.
$$
\end{corollary}
\begin{proof}
The expression for $\dot k$ follows from real part of eq. (\ref{eq-kdot}), where we take to account that $\dot k, A_{\bar\gamma}\in \mathbb{R}$. The second equation is obtained from (\ref{eq-kddot}) with $\dot k, A_{\bar\gamma}\in \mathbb{R}$ and $\frac{\partial A_{\bar\gamma}}{\partial k} = 0$ taken into account.
\end{proof}

\section{Examples}\label{sec-ex1}

\subsection{Circle with two leads and a $\delta$ coupling}  \label{subsec-ex1}
\begin{figure}
\centering
\includegraphics[height=4cm]{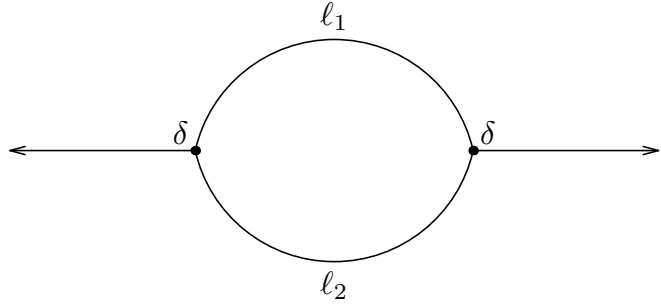}
\caption{A circle with two leads considered in subsection~\ref{subsec-ex1}.}
\label{fig10}
\end{figure}

Consider first the same example which appeared in Figure~2 of \cite{LZ}, cf. Figure~\ref{fig10}). The graph consists of two vertices and two internal edges of lengths $\ell_1(t)$, $\ell_2(t)$, where $\ell_1(0) = \ell_2(0) = \ell$, which connect the two vertices. In each vertex there is one lead attached. Unlike \cite{LZ} we consider $\delta$-coupling of the same strength $\alpha$ at both the vertices. Let the function on both internal edges be $f_1(x)$ and $f_2(x)$, on the external edges $g_1(x)$, $g_2(x)$. Then the coupling conditions are
\begin{eqnarray*}
  f_1(0) = f_2(0) = g_1(0)\,,&\quad & f_1'(0) + f_2'(0)+g_1'(0) = \alpha f_1(0)\,,\\
  f_1(\ell_1) = f_2(\ell_2) = g_2(0)\,,&\quad & -f_1'(\ell_1)-f_2'(\ell_2)+g_2'(0) = \alpha f_1(\ell_1)\,.
\end{eqnarray*}
The resonance condition can be derived as in \cite{Li3}. The effective vertex-scattering matrix is
$$
  \tilde \sigma(k) = \frac{1}{3-\frac{\alpha}{ik}} \left(\begin{array}{cc}\frac{\alpha}{ik}-1&2\\2&\frac{\alpha}{ik}-1\end{array}\right)\,,
$$
which yields
$$
  1+\frac{\left(\frac{\alpha}{ik}-1\right)^2}{\left(\frac{\alpha}{ik}-3\right)^2}\left(\mathrm{e}^{2ik\ell_1}+\mathrm{e}^{2ik\ell_2}\right)(-1)
  + 2\frac{2^2}{\left(\frac{\alpha}{ik}-3\right)^2} \mathrm{e}^{ik(\ell_1+\ell_2)} (-1) + \frac{\left(\frac{\alpha}{ik}+1\right)^2}{\left(\frac{\alpha}{ik}-3\right)^2} \mathrm{e}^{2ik(\ell_1+\ell_2)} = 0\,.
$$
This can be further rewritten by redefining the coefficients $A_{\bar\gamma}(k)$ as
\begin{eqnarray*}
  F(k(t),t) = (\alpha-3ik)^2-(\alpha-ik)^2\left(\mathrm{e}^{2ik\ell_1}+\mathrm{e}^{2ik\ell_2}\right) \\
  +8k^2 \mathrm{e}^{ik(\ell_1+\ell_2)}+(\alpha+ik)^2\mathrm{e}^{2ik(\ell_1+\ell_2)} = 0\,.
\end{eqnarray*}
Using Theorem~\ref{thm-main}, or alternatively computing the derivatives of the resonance condition at $t=0$ directly, we get
$$
  -16 \dot k k \ell (\alpha-ik)+8(\dot \ell_1+\dot \ell_2) k^2 (-\alpha+ik) = 0\,,
$$
\begin{eqnarray*}
  \ddot k 2k\ell(\alpha-ik)+\dot k k (\dot \ell_1+\dot \ell_2)(\alpha^2\ell+3k^2\ell+4\alpha+6ik\alpha\ell-6ik)+(\dot k)^2(\alpha^2\ell^2+3k^2\ell^2+4\alpha\ell-8ik\ell+6ik\ell^2\alpha)\\
  +(\ddot\ell_1+\ddot \ell_2)k^2(\alpha-ik)  +(\dot\ell_1^2+\dot \ell_2^2)k^3(k+2i\alpha)+\dot\ell_1\dot\ell_2 k^2(\alpha^2+k^2+2ik\alpha) = 0  \,.
\end{eqnarray*}
From these equations we find
$$
  \dot k = - \frac{(\dot \ell_1+\dot \ell_2)k}{2\ell}\,,
$$
\begin{eqnarray*}
   \mathrm{Re\,}\ddot k = -\frac{1}{2k\ell(\alpha^2+k^2)}\left[\dot k k(\dot\ell_1+\dot\ell_2)(\alpha^3\ell+4\alpha^2+6k^2-3k^2\alpha\ell)
   +(\dot k)^2(\alpha^3\ell^2+4\alpha^2\ell-3k^2\alpha\ell^2+8k^2\ell)+\right.\\
   \left.+(\ddot\ell_1+\ddot \ell_2)k^2(\alpha^2+k^2)
   -(\dot\ell_1^2 +\dot\ell_2^2)k^4\alpha+\dot\ell_1\dot\ell_2 k^2\alpha(\alpha^2-k^2)\right]\,.
\end{eqnarray*}
\begin{eqnarray*}
   \mathrm{Im\,}\ddot k = -\frac{1}{2k\ell(\alpha^2+k^2)}\left[\dot k k^2(\dot\ell_1+\dot\ell_2)(7\alpha^2\ell+3k^2\ell-2\alpha)\right.\\
   \left.(\dot k)^2k(7\alpha^2\ell^2+3k^2\ell^2-4\alpha\ell)+(\dot\ell_1^2+\dot\ell_2^2)k^3(2\alpha^2+k^2)+\dot\ell_1\dot\ell_2 k^3(3\alpha^2+k^2)\right]\,.
\end{eqnarray*}
These formul\ae\ determine the weak-perturbation asymptotics of the resonance trajectory in Figure~\ref{fig1}. This plot shows that unlike the case with standard coupling the knowledge of $\mathrm{Re\,}\ddot k$ is needed to approximate the resonance trajectory.
\begin{figure}
\centering
\begin{minipage}{0.95\textwidth}
\centering
\includegraphics[height=4cm]{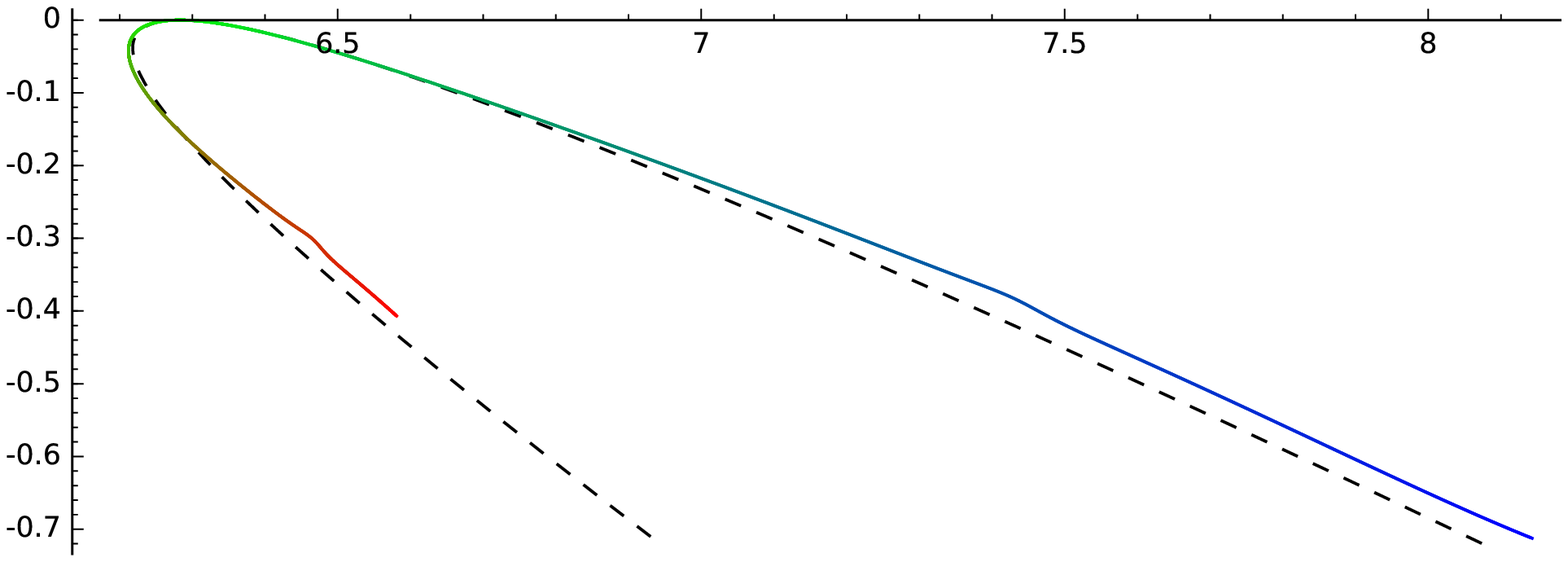}
\end{minipage}\hfill
\begin{minipage}{0.05\textwidth}
\centering
\includegraphics[height=4cm]{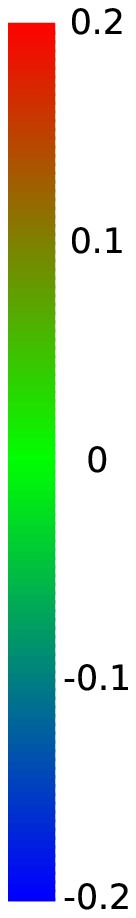}
\end{minipage}
\caption{The resonance trajectory for the graph in subsection~\ref{subsec-ex1} coming from the eigenvalue with $k_0 = 2\pi$, $\ell_1 = 1-t$, $\ell_2 = 1+2t$, $\alpha = 10$. The trajectory is shown for $t \in (-0.2, 0.2)$ and it is approximated by the dashed curve $k = k_0 + t \dot k + \frac{t^2}{2} \mathrm{Re\,}\ddot k +\frac{i t^2}{2} \mathrm{Im\,}\ddot k$ with $\dot k = -\pi$, $\mathrm{Re\,}\ddot k = 75.61$, $\mathrm{Im\,}\ddot k = -44.41$. (Color online.)}
\label{fig1}
\end{figure}

\subsection{Cross-shaped resonator}\label{subsec-ex2}

\begin{figure}
\centering
\includegraphics[height=5cm]{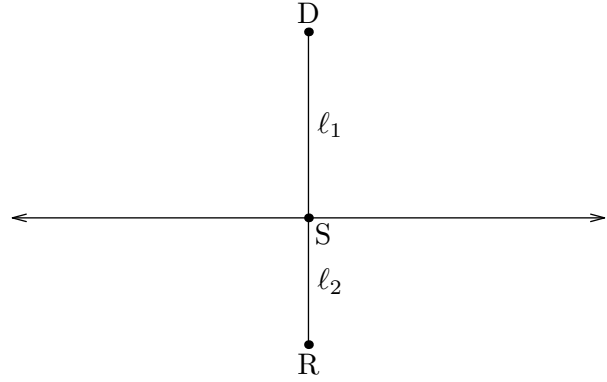}
\caption{The cross-shaped resonator considered in subsection~\ref{subsec-ex2}.}
\label{fig11}
\end{figure}

Consider next the example analyzed in subsection 4.2 of \cite{EL2}, cf. Figure~\ref{fig11}, but with different vertex conditions. The graph consists of two internal edges and two leads connected at one central vertex with the standard coupling. There is Dirichlet condition at the loose end of the internal edge of length $\ell_1$ and Robin condition $f'(0) = \alpha f(0)$ at the loose end of the second internal edge of length $\ell_2$. We choose the Ansatz $f_1(x) = a_1\sin{kx}$ at the first internal edge and $f_2 (x) = a_2 \sin{kx}+b_2\cos{kx}$ at the second one. From the Robin condition we have $a_2 k = \alpha b_2$ and hence $f_2(x) = b_2 \left(\frac{\alpha}{k}\sin{kx}+\cos{kx}\right)$. On the leads we choose the Ansatz $g_j = c_j\,\mathrm{e}^{ikx}$, $j = 1,2$. The coupling condition in the central vertex gives
\begin{eqnarray*}
  a_1\sin{k\ell_1} = b_2 \left(\frac{\alpha}{k}\sin{k\ell_2}+\cos{k\ell_2}\right) = c_1=c_2\,,\\
  -a_1k\cos{k\ell_1}-b_2\left(\alpha\cos{k\ell_2}-k\sin{k\ell_2}\right)+ik(c_1+c_2) = 0 \,.
\end{eqnarray*}
From this system of equations we obtain the resonance condition as the condition of solvability of this system.
$$
  k\cos{k\ell_1}\cos{k\ell_2}+(\alpha-2ik)\sin{k\ell_1}\cos{k\ell_2}+\alpha\cos{k\ell_1}\sin{k\ell_2}
  -(2i\alpha+k)\sin{k\ell_1}\sin{k\ell_2} = 0\,.
$$
 
\begin{figure}
\centering
\begin{minipage}{0.85\textwidth}
\centering
\includegraphics[height=8cm]{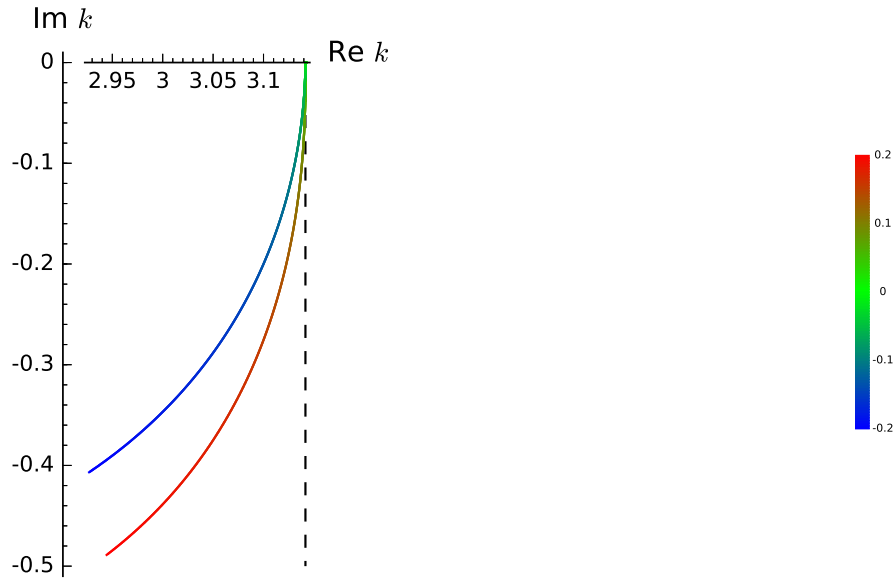}
\end{minipage}\hfill
\begin{minipage}{0.15\textwidth}
\centering
\includegraphics[height=4cm]{fig1b}
\end{minipage}
\caption{The resonance trajectory for the graph in subsection~\ref{subsec-ex2} coming from the eigenvalue with $k_0 = \pi$, $\ell_1 = 1-t$, $\ell_2 = 0.74266+t$, $\alpha = 3$. The trajectory is shown for $t \in (-0.2, 0.2)$ and it is approximated by the dashed curve $k = k_0 + t \dot k + \frac{t^2}{2} \mathrm{Re\,}\ddot k +\frac{i t^2}{2} \mathrm{Im\,}\ddot k$ with $\dot k = 0$, $\mathrm{Re\,}\ddot k = 0$, $\mathrm{Im\,}\ddot k = -20.76$. (Color online.)}
\label{fig9}
\end{figure}

The system has a real eigenvalue if the wavefunction components on the leads are zero and the components on the internal edges vanish at the central vertex. From that we obtain $\sin{k\ell_1} = 0$, and therefore $k\ell_1 = n\pi$, $n\in\mathbb{Z}$. From the second edge we have $\frac{\alpha}{k}\sin{k\ell_2}+\cos{k\ell_2} = 0$. Hence the lengths of the edges must satisfy
$$
  \ell_2 = -\frac{\ell_1}{n\pi} \mathrm{arccot\,}{\frac{\alpha\ell_1}{n\pi}}\,.
$$
The resonance trajectory is shown in Figure~\ref{fig9} for $\alpha = 3$ and the starting values $\ell_1=1$, $\ell_2 = 0.74266$, $k_0 = \pi$.

\section{High-energy asymptotics of resonances for the $\delta$ and $\delta_{\mathrm{s}}'$ coupling}

Let us now address the second problem mentioned in the introduction, namely the asymptotics of the resonances for $\delta$ and $\delta_{\mathrm{s}}'$-coupled leads in the high-energy regime. The resonances form an infinite sequence which we order with respect to the real parts of the pole positions in the ascending order; we will ask about its limiting behavior. For the sake of brevity we shall speak of the limits as the real parts tend to infinity.  

\begin{theorem}  \label{thm-delta}
Consider a graph $\Gamma$ with a $\delta$-coupling at all the vertices. Its resonances converge to the resonances of the same graph with the standard conditions as their real parts tend to infinity.
\end{theorem}
\begin{proof}
First, we consider a vertex which connects $n$ internal and $m$ external edges. We will find the effective vertex-scattering matrix. We parametrize the internal edges by $(0,\ell_j)$ and the external edges by $(0,\infty)$, in both cases zero corresponds to this vertex. We will choose the Ansatz $f_j(x) = a_j^{\mathrm{out}}\,\mathrm{e}^{ikx}+a_j^{\mathrm{in}}\,\mathrm{e}^{-ikx}$, $j = 1,\dots, n$ on the internal edges and $g_s(x) = b_s \,\mathrm{e}^{ikx}$, $s=1,\dots, m$ on the external edges. The coupling conditions yield
\begin{eqnarray*}
  a_j^{\mathrm{out}}+a_j^{\mathrm{in}} = a_i^{\mathrm{out}}+a_i^{\mathrm{in}} = b_s\,,\quad i, j=1,\dots, n\,,\quad s = 1,\dots, m\,,\\
  ik \sum_{j=1}^n (a_j^{\mathrm{out}}-a_j^{\mathrm{in}})+ik \sum_{s=1}^m b_s = \alpha (a_i^{\mathrm{out}}+a_i^{\mathrm{in}})\,,
\end{eqnarray*}
where $\alpha$ is the strength of the interaction. Substituting $b_s = a_i^{\mathrm{out}}+a_i^{\mathrm{in}}$ and $a_j^{\mathrm{out}}=a_i^{\mathrm{out}}+a_i^{\mathrm{in}} -a_j^{\mathrm{in}}$ to the second equation we obtain
$$
  (n +m)(a_i^{\mathrm{out}}+a_i^{\mathrm{in}}) -2\sum_{j=1}^n  a_j^{\mathrm{in}} = \frac{\alpha}{ik}(a_i^{\mathrm{out}}+a_i^{\mathrm{in}})\,,
$$
and hence
$$
  a_i^{\mathrm{out}} = \frac{2}{n+m-\frac{\alpha}{ik}}\sum_{j=1}^n  a_j^{\mathrm{in}} -a_i^{\mathrm{in}}\,,
$$
which gives the effective vertex-scattering matrix
$$
  \tilde \sigma_{\delta} (k) =  \frac{2}{n+m-\frac{\alpha}{ik}}J-I \,\longrightarrow\, \frac{2}{n+m}J-I = \tilde \sigma_{\mathrm{st}} (k)
$$
as $|k|\to\infty$. Hence in the limit, where $k$ goes to infinity in modulus the effective vertex-scattering matrix for $\delta$-coupling tends to the effective vertex-scattering matrix for the standard condition.

Since the coefficient $A_{\bar\gamma}$ in the pseudo-orbit expansion is a product of the entries of the effective vertex-scattering matrices, we have
$$
  A_{\bar\gamma,\delta} (k) = A_{\bar\gamma,\mathrm{st}} (k) + \mathcal{O}\left(|k|^{-1}\right)\,,
$$
where $A_{\bar\gamma,\delta}$ is the coefficient for the $\delta$-coupling and $A_{\bar\gamma,\mathrm{st}}$ is the coefficient for the respective standard condition. Consequently, by Theorem~4 in \cite{La} the resonances coming from the $\delta$-coupling coincide asymptotically with those of the same graph with the standard coupling conditions.
\end{proof}

\begin{theorem}   \label{thm-deltaprime}
The resonances of the graph with a $\delta_{\mathrm{s}}'$ coupling conditions at the vertices converge to the eigenvalues of the graph with Neumann (decoupled) conditions as their real parts tend to infinity.
\end{theorem}
\begin{proof}
The proof is similar to that of Theorem~\ref{thm-delta}. First, we find the effective vertex-scattering matrix for $\delta_{\mathrm{s}}'$-condition. We use the same Ansatz as there; from the coupling condition we obtain
\begin{eqnarray*}
  a_j^{\mathrm{out}}-a_j^{\mathrm{in}} = a_i^{\mathrm{out}}-a_i^{\mathrm{in}} = b_s\,,\quad i, j=1,\dots, n\,,\quad s = 1,\dots, m\,,\\
   \sum_{j=1}^n (a_j^{\mathrm{out}}+a_j^{\mathrm{in}})+ \sum_{s=1}^m b_s = ik\beta (a_i^{\mathrm{out}}-a_i^{\mathrm{in}}) \,,
\end{eqnarray*}
where $\beta$ is the strength of the interaction. Substituting $b_s = a_i^{\mathrm{out}}-a_i^{\mathrm{in}}$ and $a_j^{\mathrm{out}}=a_i^{\mathrm{out}}-a_i^{\mathrm{in}} +a_j^{\mathrm{in}}$ to the second equation we obtain
$$
  ik\beta (a_i^{\mathrm{out}}-a_i^{\mathrm{in}}) = 2 \sum_{j=1}^n  a_j^{\mathrm{in}} + (n+m)(a_i^{\mathrm{out}}-a_i^{\mathrm{in}})\,,
$$
and hence
$$
   a_i^{\mathrm{out}} = a_i^{\mathrm{in}} +\frac{2}{ik\beta -n-m}\sum_{j=1}^n a_j^{\mathrm{in}}\,.
$$
Therefore, the effective vertex-scattering matrix is
$$
  \tilde \sigma_{\delta_{\mathrm{s}}'} = \frac{2}{ik\beta -n-m}J+I \,\longrightarrow\, I =  \tilde \sigma_{\mathrm{N}}
$$
as $|k|\to\infty$, where $\tilde \sigma_{\mathrm{N}}$ is the effective vertex-scattering matrix for (decoupled) Neumann condition. The coefficients in the pseudo-orbit expansion $A_{\bar\gamma,\delta_{\mathrm{s}}'}$ corresponding to the $\delta_{\mathrm{s}}'$-conditions and $A_{\bar\gamma,\mathrm{N}}$ for the Neumann condition satisfy
$$
  A_{\bar\gamma,\delta_{\mathrm{s}}'} (k) = A_{\bar\gamma,\mathrm{N}} (k) + \mathcal{O}\left(|k|^{-1}\right)\,.
$$
Invoking again Theorem~4 of \cite{La} we conclude that the resonances of $\delta_{\mathrm{s}}'$-coupling asymptotically coincide with those of the same graph with Neumann conditions.
\end{proof}

In fact, the above argument can be extended to the situations where the leads are attached by a $\delta_{\mathrm{s}}'$-coupling and the other vertex conditions in the graph are arbitrary self-adjoint.

\begin{theorem}
In the described situation the resonances of the graph satisfy
$$
  \mathrm{Im\,}k \to 0 \quad \mathrm{as}\quad |k|\to \infty\,.
$$
\end{theorem}
\begin{proof}
By the same reasoning as in the proof of Theorem~\ref{thm-deltaprime} one can verify that the resonances of $\Gamma$ converge as $|k|\to \infty$ to resonances of the same graph with the $\delta_{\mathrm{s}}'$-coupling replaced by the (decoupled) Neumann conditions. Since all the leads are decoupled from the rest of the graph with Neumann condition, the solutions to the resonances condition for the inner part of the graph are eigenvalues with $\mathrm{Im\,}k = 0$, and consequently all the resonances of $\Gamma$ approach the real axis in the asymptotic regime.
\end{proof}

To conclude this section we find the convergence rate of the resonance poles for a particular graph class. We note that a similar behavior was observed for resonances of the generalized Winter model discussed in \cite{EF}.

\begin{theorem}
Let the internal part of $\Gamma$ be equilateral with all the internal edges of lengths $\ell_0$ and let $\Gamma$ have at least one lead. Assume further that there is the same $\delta_{\mathrm{s}}'$-coupling at all the vertices. Denote $k_{0n} = n\pi/\ell_0$, then the resonances $k_n^2$ satisfy
$$
  \mathrm{Im\,}k_n = \mathcal{O}\left((\mathrm{Re\,}k_n)^{-2}\right)\,,\quad \mathrm{Re\,}(k_n-k_{0n}) = \mathcal{O}\left(({\mathrm{Re\,}k_n})^{-1}\right)
$$
as $\mathrm{Re\,}k_n \to \infty$.
\end{theorem}
\begin{proof}
The effective vertex scattering at the $s$-th vertex is $\tilde \sigma(k) = p_s J +I$ with $p_s = \frac{2}{ik\beta-n_s-m_s}$, where $n_s$ and $m_s$ is the number of internal and external edges in this vertex, respectively. We have $p_s = -\frac{2i}{\beta \mathrm{Re\,}k} + \mathcal{O}((\mathrm{Re\,}k)^{-2})$. We can see that the irreducible pseudo-orbits which contain at least one non-diagonal term of $\tilde \sigma$ correspond to the resonance condition with the terms of order $\mathcal{O}((\mathrm{Re\,}k)^{-2})$ or smaller. Hence the resonance condition is
$$
  0 = F(k) = \sum_{\bar\gamma \in V} \left(1+\sum_{v_s\in \bar\gamma} p_s\right)\mathrm{e}^{ik\ell_{\bar\gamma}}(-1)^{m_{\bar\gamma}} +\mathcal{O}\left((\mathrm{Re\,}k)^{-2}\right)\,,
$$
where $v_s$ is the $s$-th vertex and $V$ is the set of irreducible pseudo-orbits which contain only periodic orbits on two bonds (e.g., periodic orbits which use only diagonal terms of $\tilde\sigma$). Using the Taylor expansion around $k_{0n}=n\pi/\ell_0$ we find (having dropped the subscript $n$)
\begin{eqnarray}
 0 = F(k) = F(k_0)+ (k-k_0) \left.\frac{\partial F}{\partial k}\right|_{k_0} + \mathcal{O}((k-k_0)^2) = \sum_{\bar\gamma\in V}\mathrm{e}^{ik_0\ell_{\bar\gamma}}(-1)^{m_{\bar\gamma}}-
 \frac{2i}{\beta \mathrm{Re\,}k} \sum_{\bar \gamma\in V} \ell_{\bar\gamma} \mathrm{e}^{ik_0\ell_{\bar\gamma}}(-1)^{m_{\bar\gamma}} \nonumber\\
 + [ i \mathrm{Re\,}(k-k_0)-\mathrm{Im\,}k] \sum_{\bar \gamma\in V} \ell_{\bar\gamma} \mathrm{e}^{ik_0\ell_{\bar\gamma}}(-1)^{m_{\bar\gamma}}\left(1+\mathcal{O}\left((\mathrm{Re\,}k)^{-1}\right)\right) + \mathcal{O}\left((\mathrm{Re\,}k)^{-2}\right) + \mathcal{O}((k-k_0)^2)\,.\label{eq-compare}
\end{eqnarray}
The first term $\sum_{\bar\gamma\in V}\mathrm{e}^{ik_0\ell_{\bar\gamma}}(-1)^{m_{\bar\gamma}}$ is zero since $\mathrm{e}^{ik_0\ell_{\bar\gamma}} = 1$ and $\sum_{j=0}^{N} {N\choose j} (-1)^j = 0$. Comparing the imaginary and real part of equation~(\ref{eq-compare}) using the fact that $ \mathrm{e}^{ik_0\ell_{\bar\gamma}} \in\mathbb{R}$ we obtain the sought claim.
\end{proof}

\section{Examples}

\subsection{Loop with two halflines and a $\delta$-coupling} \label{sec-delta}

We consider the same graph as in subsection~\ref{subsec-ex1}, but now we suppose that there are $\delta$-couplings of strengths $\alpha_1$ and $\alpha_2$ at its vertices. The resonance condition is in this case
$$
  (\alpha_1-ik)(\alpha_2-ik)\sin{k\ell_1}\sin{k\ell_2}-4k^2\sin^2{\frac{k(\ell_1+\ell_2)}{2}}
  +k(\alpha_1+\alpha_2-2ik)\sin{k(\ell_1+\ell_2)} = 0\,.
$$
Positions of the resonances of this graph are shown in Figures~\ref{fig2} and \ref{fig3}; for simplicity, we use $\alpha_1 = \alpha_2$.

\begin{figure}
\centering
\includegraphics[height=5cm]{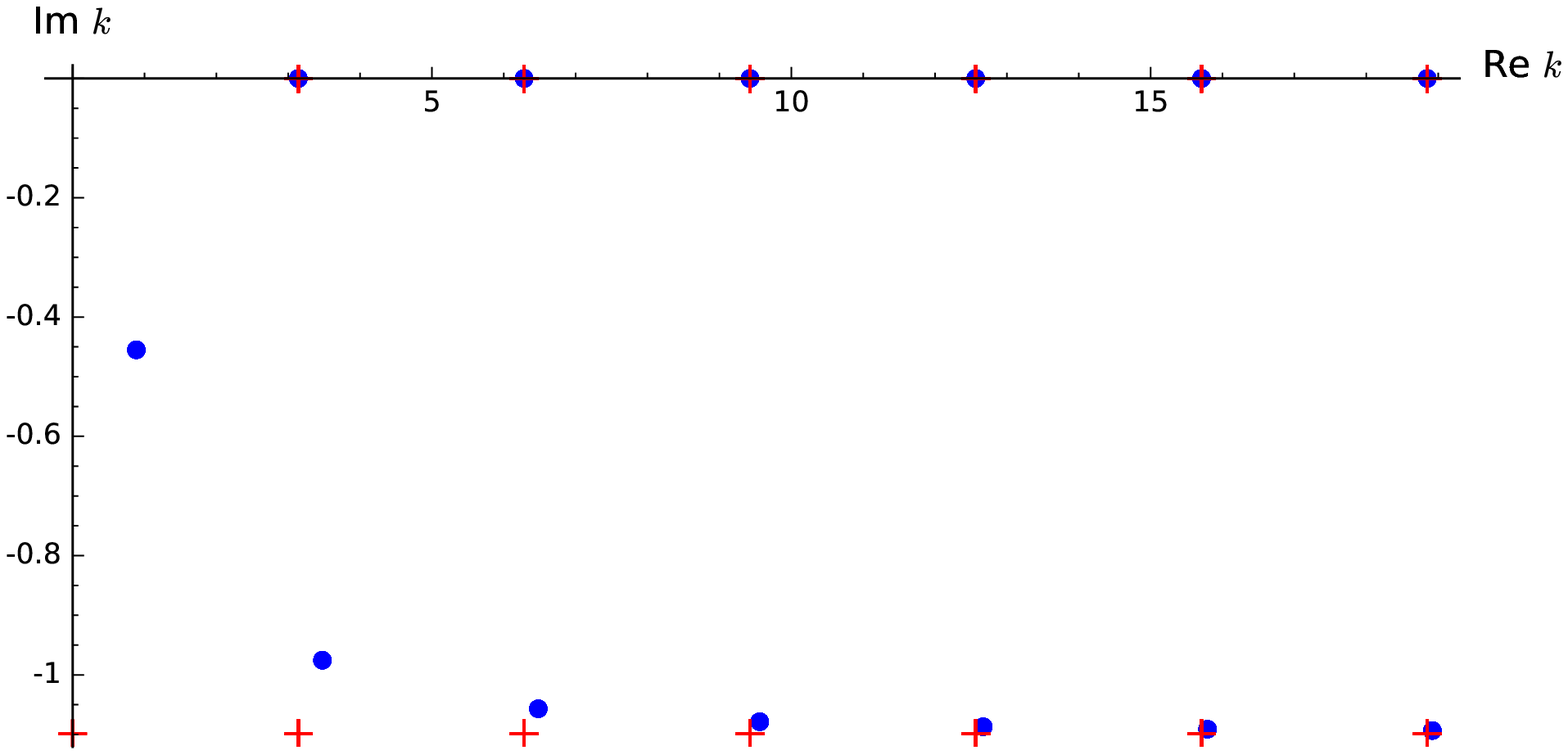}
\caption{Illustration to example in subsection~\ref{sec-delta} with the parameters $\ell_1 = 1$; $\ell_2 = 1$; $\alpha_1 = 1$; $\alpha_2 = 1$. Resonances for $\delta$-condition denoted by blue dots, resonances for standard condition by red crosses. (Color online.)}
\label{fig2}
\end{figure}

\begin{figure}
\centering
\includegraphics[height=5cm]{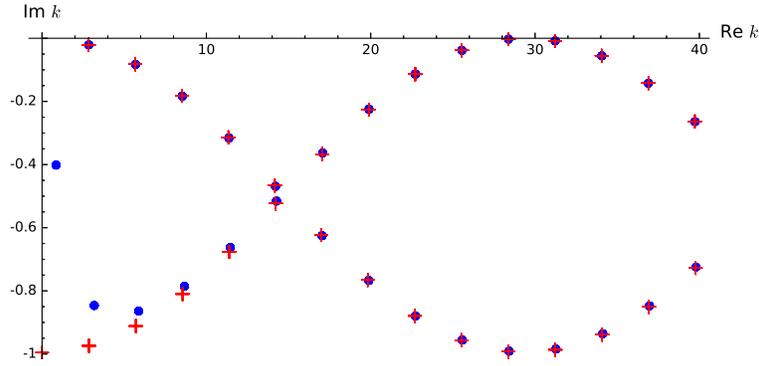}
\caption{Illustration to example in subsection~\ref{sec-delta} with the parameters $\ell_1 = 1$; $\ell_2 = 1.2137$; $\alpha_1 = 1$; $\alpha_2 = 1$. Resonances for $\delta$-condition denoted by blue dots, resonances for standard condition by red crosses. (Color online.)}
\label{fig3}
\end{figure}

\subsection{Loop with two halflines and a $\delta_{\mathrm{s}}'$-coupling} \label{sec-deltaprime}

Let us consider the same graph as in subsections~\ref{subsec-ex1} and~\ref{sec-delta} but now with $\delta_{\mathrm{s}}'$-conditions in both the vertices. The strengths of these interactions are $\beta_1$ and $\beta_2$. We choose the Ansatz $f_j (x) = a_j\,\mathrm{e}^{ikx}+b_j\,\mathrm{e}^{-ikx}$, $j = 1,2$, on both internal edges with $x=0$ in the left vertex, while on the leads we have  $g_j(x) = c_j\,\mathrm{e}^{ikx}$, $j = 1,2$. This yields
\begin{eqnarray*}
a_1-b_1 = a_2-b_2 = c_1 \,,&\quad &
ik\beta_1c_1 = a_1+b_1+a_2+b_2+c_1\,,\\
-a_1\,\mathrm{e}^{ik\ell_1}+b_1\,\mathrm{e}^{-ik\ell_1} = -a_2\,\mathrm{e}^{ik\ell_2}+b_2\,\mathrm{e}^{-ik\ell_2} = c_2\,,&\quad &
ik\beta_2 c_2 = a_1 \,\mathrm{e}^{ik\ell_1}+b_1\,\mathrm{e}^{-ik\ell_1}+a_2 \,\mathrm{e}^{ik\ell_2}+b_2\,\mathrm{e}^{-ik\ell_2}+c_2
\end{eqnarray*}
and the resonance condition is equivalent to the requirement of solvability of the above system,
$$
  [(\beta_1+\beta_2)k+2i]\sin{k(\ell_1+\ell_2)}+2(1-\cos{k\ell_1}\cos{k\ell_2})
  +(3-\beta_1\beta_2k^2-ik(\beta_1+\beta_2))\sin{k\ell_1}\sin{k\ell_2} = 0 \,.
$$
The corresponding resonance positions are shown in Figures~\ref{fig4}, \ref{fig5} and \ref{fig6}.

\begin{figure}
\centering
\includegraphics[height=5cm]{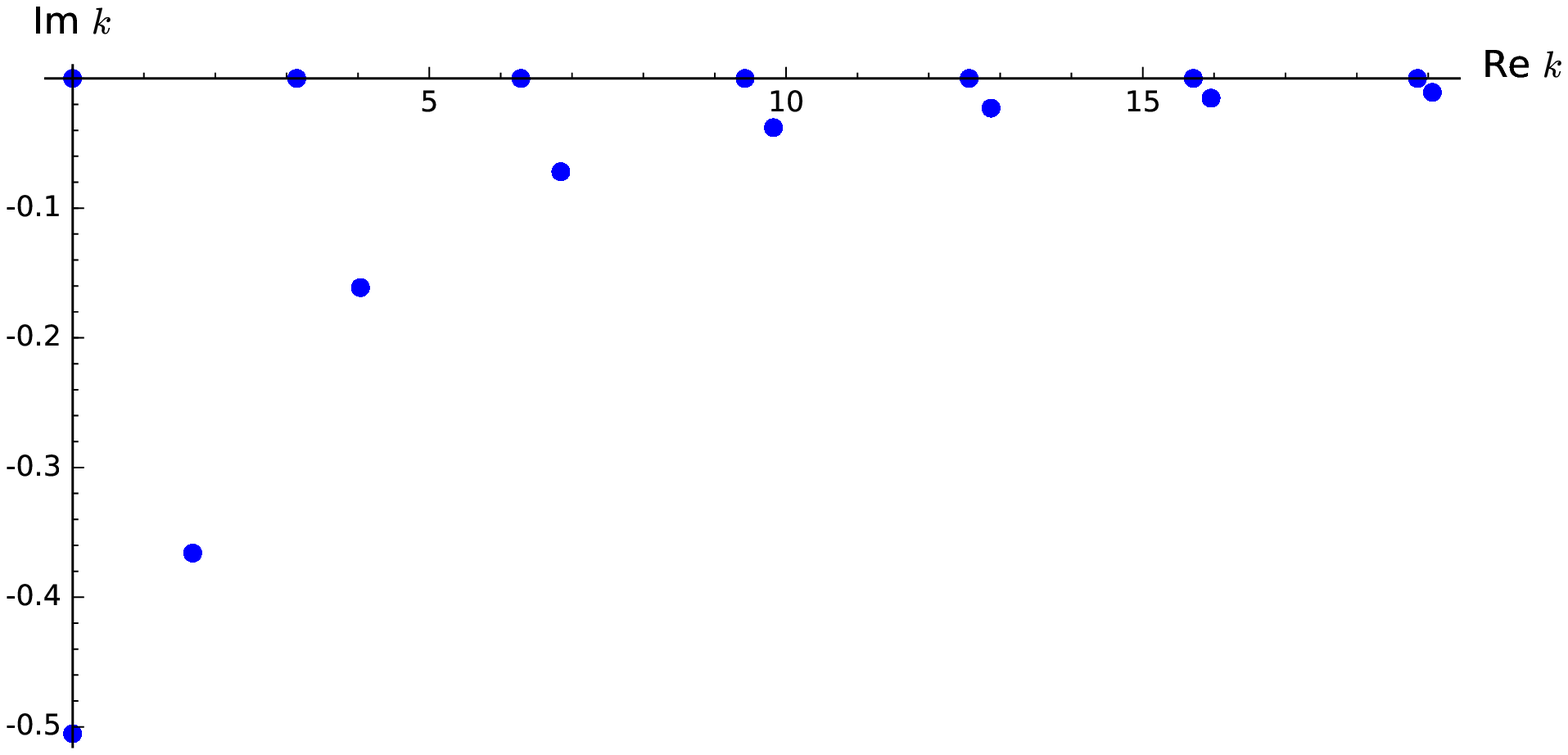}
\caption{Illustration to example in subsection~\ref{sec-deltaprime} with the parameters $\ell_1 = 1$; $\ell_2 = 1$; $\beta_1 = 1$; $\beta_2 = 1$.}
\label{fig4}
\end{figure}

\begin{figure}
\centering
\includegraphics[height=5cm]{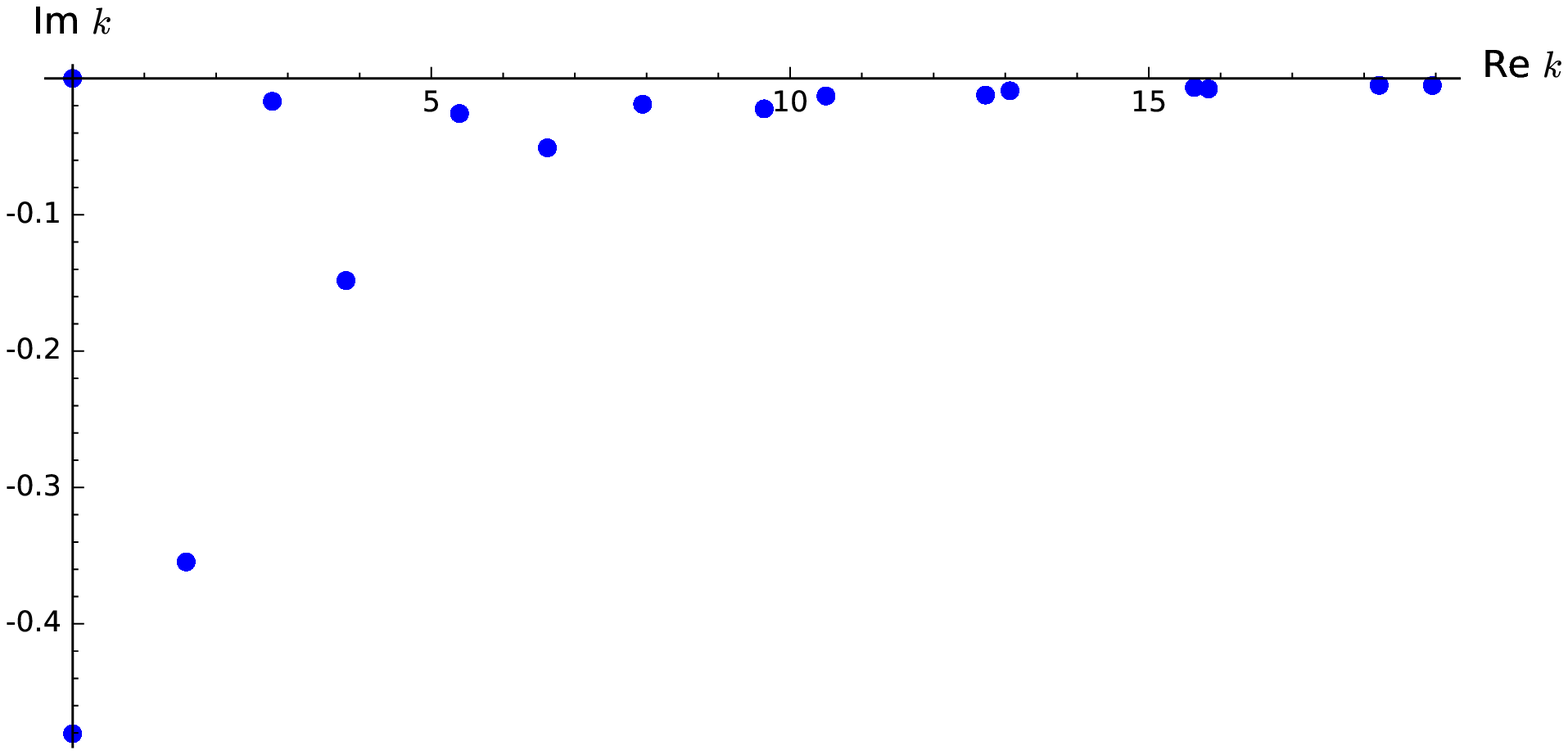}
\caption{Illustration to example in subsection~\ref{sec-deltaprime} with the parameters $\ell_1 = 1$; $\ell_2 = 1.2137$; $\beta_1 = 1$; $\beta_2 = 1$.}
\label{fig5}
\end{figure}

\begin{figure}
\centering
\includegraphics[height=5cm]{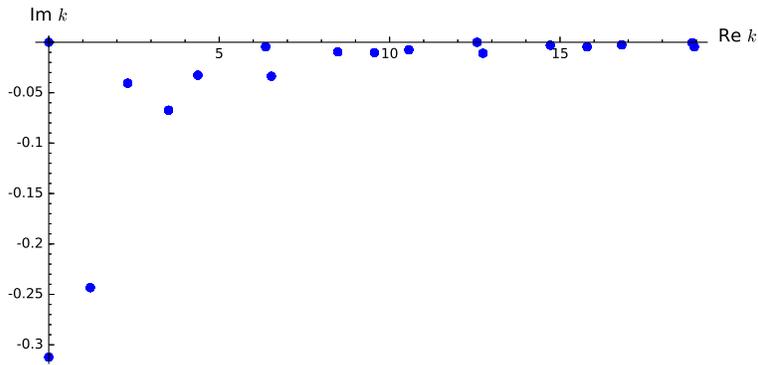}
\caption{Illustration to example in subsection~\ref{sec-deltaprime} with the parameters $\ell_1 = 1$; $\ell_2 = 1.5$; $\beta_1 = 1$; $\beta_2 = 3$.}
\label{fig6}
\end{figure}

\subsection{Loop with two leads and a combination of $\delta$ and $\delta_{\mathrm{s}}'$-couplings} \label{sec-combin}

We consider the same graph again with the $\delta$-coupling of the strength $\alpha$ in the left vertex and $\delta_{\mathrm{s}}'$-coupling of the strength $\beta$ in the right vertex. We choose the same Ansatz as in subsection~\ref{sec-deltaprime}. The coupling conditions yield
\begin{eqnarray*}
a_1+b_1 = a_2+b_2 = c_1\,,&\quad &
ik (a_1-b_1+a_2-b_2+c_1) = \alpha c_1\,,\\
c_2+a_1\,\mathrm{e}^{ik\ell_1}-b_1 \,\mathrm{e}^{-ik\ell_1} = c_2+a_2\,\mathrm{e}^{ik\ell_2}-b_2 \,\mathrm{e}^{-ik\ell_2} = 0\,,&\quad &
\beta ik c_2 = a_1\,\mathrm{e}^{ik\ell_1}+b_1 \,\mathrm{e}^{-ik\ell_1}+a_2\,\mathrm{e}^{ik\ell_2}+b_2 \,\mathrm{e}^{-ik\ell_2}+c_2
\end{eqnarray*}
and the resonance condition is
$$
  (\beta k^2+ik\alpha\beta +3ik-\alpha)\cos{k\ell_1}\cos{k\ell_2}+(-i\beta k^2+i\alpha+2k)\sin{k(\ell_1+\ell_2)}
  -2ik\sin{k\ell_1}\sin{k\ell_2}+2ik = 0 \,.
$$
Positions of the resonances are shown in Figures~\ref{fig7} and \ref{fig8}.

\begin{figure}
\centering
\includegraphics[height=5cm]{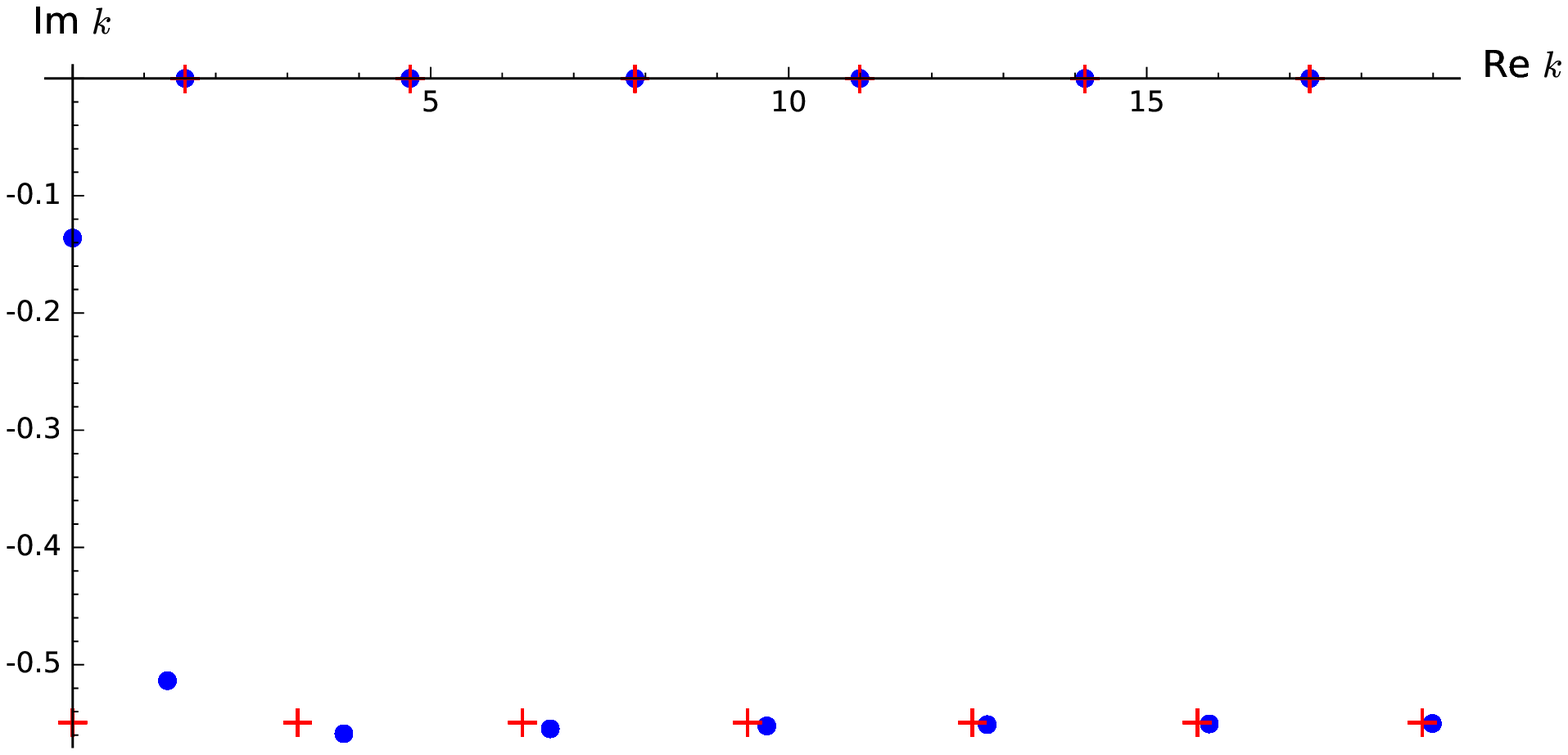}
\caption{Illustration to example in subsection~\ref{sec-combin} with the parameters $\ell_1 = 1$; $\ell_2 = 1$; $\alpha = 1$; $\beta = 1$. Resonances for $\delta$-condition and $\delta'_{\mathrm{s}}$-condition denoted by blue dots, resonances for $\delta$ replaced by standard condition and $\delta'_\mathrm{s}$ replaced by Neumann condition denoted by red crosses. (Color online.)}
\label{fig7}
\end{figure}

\begin{figure}
\centering
\includegraphics[height=5cm]{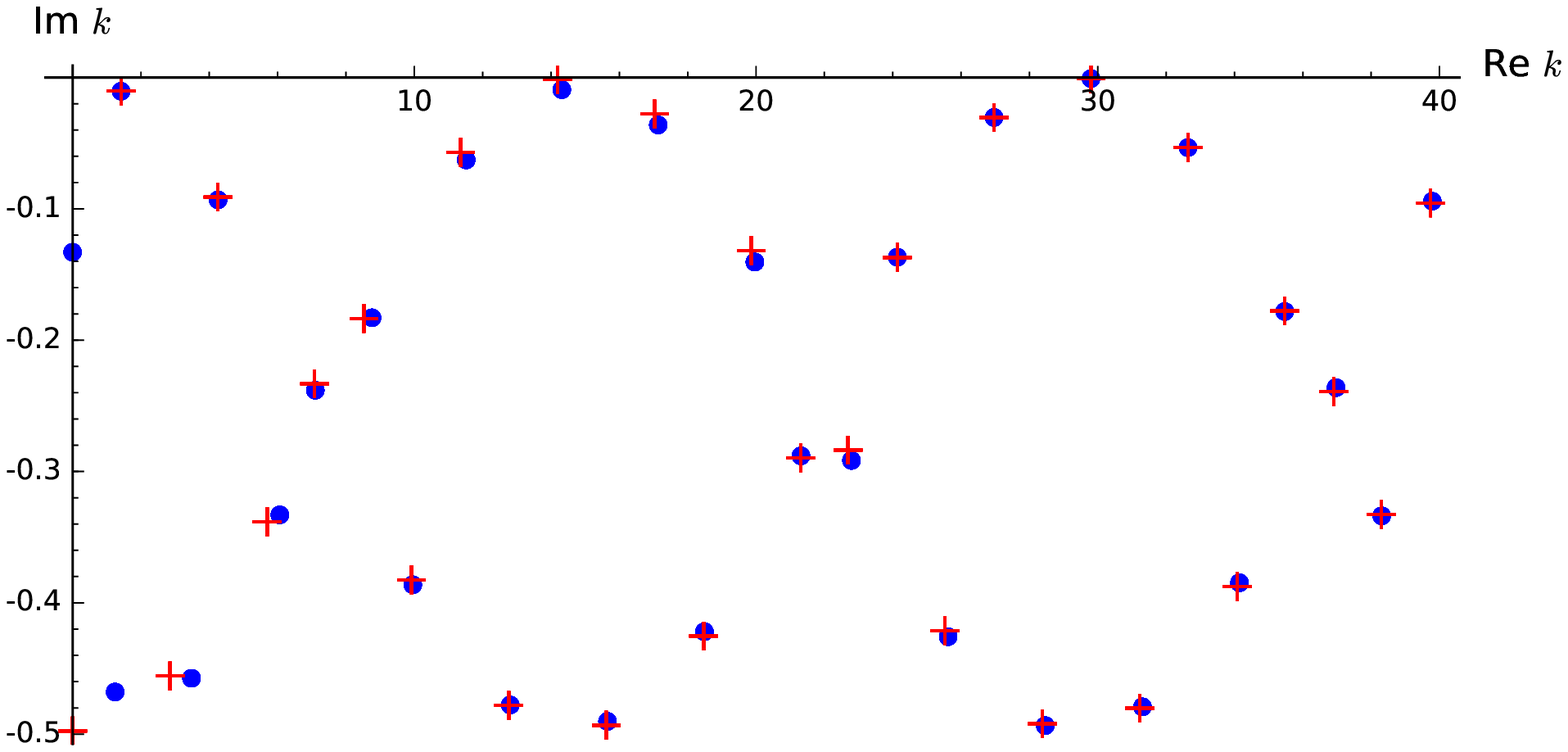}
\caption{Illustration to example in subsection~\ref{sec-combin} with the parameters  $\ell_1 = 1$; $\ell_2 = 1.2137$; $\alpha = 1$; $\beta = 1$. Resonances for $\delta$-condition and $\delta'_{\mathrm{s}}$-condition denoted by blue dots, resonances for $\delta$ replaced by standard condition and $\delta'_\mathrm{s}$ replaced by Neumann condition denoted by red crosses. (Color online.)}
\label{fig8}
\end{figure}


%
%

%



\end{document}